\documentclass{llncs}

\usepackage{amsmath}
\usepackage{amsfonts}
\usepackage{tikz}
 
\PassOptionsToPackage{hyphens}{url}\usepackage{hyperref}
\usepackage{mathtools} %for \coloneqq

\usepackage[
lambda,
probability,
sets,
operators,
]{cryptocode}
\createprocedureblock{procb}{center ,boxed}{}{}{}

\newcommand{\mmod}{\!\!\mod}
\newcommand{\inv}[1]{{#1}^{-1}}
\newcommand{\pcinput}{\textbf{input: }}
\newcommand{\pcoutput}{\textbf{output: }}
\renewcommand{\vector}[1]{\sequence{#1}}
\newcommand{\etal}{\textit{et al.\@}}

\usetikzlibrary{backgrounds}
\usetikzlibrary{arrows}
\usetikzlibrary{shapes,shapes.geometric,shapes.misc}
\pgfdeclarelayer{edgelayer}
\pgfdeclarelayer{nodelayer}
\pgfsetlayers{background,edgelayer,nodelayer,main}

\pgfkeys{/tikz/tikzit fill/.initial=0}
\pgfkeys{/tikz/tikzit draw/.initial=0}
\pgfkeys{/tikz/tikzit shape/.initial=0}
\pgfkeys{/tikz/tikzit category/.initial=0}

\tikzstyle{none}=[inner sep=0mm]

\title{Delay Encryption by Cubing}
\author{Ivo Maffei \and A. W. Roscoe}
\institute{Department of Computer Science, University of Oxford}

\begin{document}

\maketitle

\begin{abstract}
Delay Encryption (often called Timed-Release Encryption) is a scheme in which a message is sent into the future by ensuring its confidentiality only for a given amount of time.
We propose a new scheme based on a novel time-lock puzzle.
This puzzle relies on the assumption that repeated squaring is an inherently sequential process.
We perform an extensive and practical analysis of many classical and quantum attacks on our scheme and conclude that it is secure given some precautions.
\end{abstract}

\section{Introduction}
Timed-Release Encryption (TRE) is a method by which one can encrypt a message only for some amount of time.
This relatively little studied encryption model has a wide range of applications spanning from online gambling to e-voting systems and internet protocols.
In 1996, Rivest \etal \cite{Rivest} proposed a few ways in which the problem of sending messages into the future could be solved.
Since then, the literature has focused on only two possible approaches:
\begin{enumerate}
\item Time-lock puzzles: computational puzzles that are guaranteed to require some given amount sequential of work.
\item Trusted third parties: interactive protocols that take advantage of the presence of a TTP.
\end{enumerate}
The literature on examples of the second approach is abundant and includes ones based on synthetic TTPs such as blockchains \cite{Lai}. 
On the other hand, the time-lock puzzle proposed by Rivest \textit{et al.}\ seems to be the only accepted example of the first category.
The main disadvantage of this method  is that it relies on the hardness of integer factorisation: an unreasonable assumption in a post-quantum world.
The recent interest that time-related encryption schemes are getting makes the need for a secure TRE scheme urgent.

In this paper, we propose a time-lock puzzle based on the sequentiality of repeated squaring which  appears to be quantum-resistant.
This idea for TRE was discussed in passing by Roscoe \cite{Roscoe} and here we analyse his idea further in order to create a concrete and secure scheme. 
Our primary aim is to achieve secure delays of  a few seconds or minutes that would allow this scheme to be adopted in internet protocols.
The disparity in computational power between a resourceful adversary and an average user makes our scheme impractical for long delays.
Therefore, we believe that the general term ``timed-release encryption'' can be misleading and we resort to the idea of ``delay encryption'' as originally intended by Roscoe \cite{Roscoe} rather than the newer definition by Burdges and De Feo \cite{DeFeo}.
Short delays  are enough to guarantee the security of protocols such as the partially-fair computation by Couteau \etal\cite{Geoffroy}, where delay encryption plays a crucial role. Specifically, the protocol in \cite{Geoffroy} betters the known upper bound  for fair exchange of \cite{Gordon} for protocols where no delay encryption was permitted. Given this result, we were motivated to ensure that delay encryption as used in \cite{Geoffroy} and similar protocols is soundly based.
Long-term delays via time-lock puzzles are very impractical because they force the receiver to do a very expensive computation.
Improvements in algorithmic theory and hardware make the estimate of years-long delays very imprecise.

Although the use of the cubing operation is not novel in the context of time-related cryptography, its use for timed-release encryption is.
Many other researchers (e.g.\@ \cite{vdf,Dwork,Ephraim20a,Jerschow,pietrzak}) proposed similar constructions based on repeated squaring in research areas related to TRE. 
However, most of these works are based on the time-lock puzzle by Rivest \etal \cite{Rivest} and they all lack a key component of our research: a practical analysis of the scheme.
In particular, as we will see in Sections \ref{sec:attacks} and \ref{sec:pq},  an adversary with large enough computational power can easily break the security of our scheme as well as the scheme by Rivest \textit{et al}.
However, this turns out impractical with the current state of technology.
With this paper, we aim at starting a discussion on practical TRE schemes that will withstand attackers equipped quantum computers.

To motivate the security of our scheme, we performed an extensive and practical analysis of the sequentiality of modular exponentiation.
This part of our work is interesting in its own right since it naturally translates to the context of the RSA-based time-lock puzzle by Rivest \etal\cite{Rivest}. 
In this study, we point out that a resourceful malicious entity can obtain a considerable speedup in computing the exponentiation as the modulus grows larger.
This issue can be mitigated by chaining together shorter and safer delays. Although this is a common technique in the construction of Verifiable Delay Functions \cite{vdf}, it has not been applied to the field of TRE and time-lock puzzles.
Composing multiple delays into a longer one is not trivial because one must ensure the sequentiality of the end result.
In Section \ref{sec:chain} we will propose a few secure options. 
Note that the use of chaining does not increase the ratio between decryption time and encryption time. 
As a result, chaining is relatively inefficient, but it suits our use in cryptographic protocols such as those by Couteau \etal\cite{Geoffroy} or by Roscoe and Ryan \cite{Ryan}.

In Section \ref{sec:tre}, we outline the encryption scheme proposed and analyse what parameters should be used to make it secure.
We then proceed to analyse ways in which our scheme can be extended to longer delays via chaining in Section \ref{sec:chain}.
In Section \ref{sec:attacks}, we perform an extensive analysis of possible classical attacks on the system. 
In Section \ref{sec:pq}, we look at the additional attacks that a quantum computer could carry against our scheme.

\section{Delay by cubing}
\label{sec:tre}
In this section,s we propose a time-lock puzzle based on the empirical hypothesis that repeated squaring cannot be parallelised.
In particular, the underlying conjecture is given by the following definition.
\begin{conjecture}[Sequentiality]\label{seqAssump}
Given $x, y, p$ with $x, y \in \mathbb Z_p$ and $p$ prime, computing $x^y \mod p$ requires computing at least $\lfloor \log y \rfloor$ modular multiplications sequentially.
\end{conjecture}
In the above conjecture, as in the rest of this article, the logarithm is taken with base 2.

\subsection{Preparation phase}
The preparation phase is equivalent to the key-generation step in a traditional scheme.
The result of this process is a tuple  of parameters $(p, b, T)$ which should be treated as a public key.
\begin{figure}[h]
\procb{Puzzle setup}{
\textbf{input: } T, \lambda\\
\textbf{output: } (p, b, T)\\
p \gets 	\text{a large safe prime with } \floor {\log p} = \lambda T\\
b \gets \frac 1 3 \left( 1 + 2 \left(p-1\right)\right)
}
\caption{Preparation phase.}
\label{fig:treSetup}
\end{figure}

Figure \ref{fig:treSetup} shows the pseudocode for the setup phase.
The input $\lambda$ represents a security parameter and $T$ indicates the least amount of time required to solve the puzzle.
In particular, $\lambda$ represents the conjectured speed (number of modular squarings per second) that a resourceful adversary could obtain.
The prime $p$ is picked to be a ``safe prime'' (i.e.\@ $\frac {p-1} 2$ is also prime, often called a Sophie Germain prime) in order to reduce possible attacks on the system.
See Section \ref{sec:safeModulus} for a deeper discussion on the reasons for this choice.
Note also that if $p$ is a safe prime, then the cubing operation is invertible in $\mathbb Z_p$.
The number $b$ is picked to be the smallest integer satisfying $x^{3b} \equiv x \mod p$ for all $x \in \ZZ_p$.

\subsection{Encryption and Decryption}\label{sec:enc}
Given a message $m \in \mathbb{Z}_p$, its encryption is $m^3 \mmod p$.
To decrypt a ciphertext $c$ is enough to compute $c^b \mmod p$.

\begin{figure}[h]
\begin{minipage}{0.45\textwidth}
\procb{Encryption}{
\textbf{input: } m, p, \lambda\\
\textbf{output: } c\\
s \sample \bin^{\lambda}\\
x \gets \textsf{pad}\left( m, s\right)\\
c \gets x^3 \mod p
}
\caption{Encryption process.}\label{fig:enc}
\end{minipage}\hfill
\begin{minipage}{0.45\textwidth}
\procb{Decryption}{
\textbf{input: } c, p, b\\
\textbf{output: } m\\
x \gets c^b \mod p\\
m \gets \textsf{pad}^{-1}\left(x\right)
}
\caption{Decryption process.}\label{fig:dec}
\end{minipage}
\end{figure}

Figure \ref{fig:enc} shows the full algorithm for the delay of a message.
In Section \ref{sec:randPadding} we will discuss in detail the purpose of $s$, in short $\lambda$ can be regarded as the key length for the encryption scheme.
In the proposed algorithm, we use a generic padding scheme \textsf{pad} which is assumed to satisfy the following conditions:
\begin{enumerate}
\item It must allow the use of  a random seed of variable length.
\item The result $x$ must satisfy $x \gg \sqrt[3] p$.
\item Removal of the padding requires no knowledge of the random seed used.
\end{enumerate}

Figure \ref{fig:dec} illustrates the decryption algorithm.
Decryption works because $3b \equiv 1 \mod (p-1)$ and for all $x \in \mathbb Z_p$ $x^{p-1} \equiv 1 \mod p$.
Therefore $x^{3b} \equiv x^{2(p-1)} x \equiv x \mod p$. 
Note that one can decrypt $x^3 \mod p$ by using another number $a$ in place of $b$ as long as  $3a\equiv 1 \mod \text{ord}(x)$.
In Section \ref{sec:safeModulus} we explain how to ensure that there is no convenience in using an exponent different from $b$.

\subsection{Security analysis}
The proposed scheme works as a delay encryption because the decryption process has time complexity $O((\log p)^2 \log \log p)$ while the encryption process only $O(\log p \log \log p)$.
However, this holds true only assuming that repeated squaring is an inherently sequential process. 
This assumption has been widely studied since the first proposal of a time-lock puzzle by Rivest \etal \cite{Rivest} and today it is still widely used.

On the practical side, we tested some timings on an Intel i7-6920HQ quad-core with Hyper-Threading enabled. 
The CPU base clock speed is 2.9GHz (up to 3.8GHz using Turbo Boost).
The software used took advantage of the open-source library GMP \cite{GMP} to compute the exponentiations.
The GMP library focuses on high performance thanks to its low-level assembly implementations of many primitives. 
Moreover, the GMP community invests a significant amount of effort in producing code that is tailored for many processor pipelines.
The C code was compiled with Apple clang version 12.0.0 and optimisation flag \texttt{-Ofast}.
Using the 70034-bit prime $p = 2566851867 \times 2^{70002} -1$ and picking random messages, we measured that, discounting padding costs, decryption takes roughly 19s while encryption requires only 0.8ms. On CPUs i9-9880H and i9-9980HK running Windows 10, we found that the decryption time is roughly 20s.
This shows that it is possible to utilise our encryption scheme to delay a message efficiently.
When making concrete computations, we will use the above prime. 
The resulting figures are only indicative, but we expect that the analysis in this paper would substantially apply to any prime in the range of 25000-250000 bits.
Our concentration on this single prime is somewhat excused by Section \ref{sec:chain}, where we will show how the delay created by a single cubing puzzle can be multiplied by a small integer factor without changing the prime. 
However, we do not claim that this choice is optimal nor that different applications should not use different primes.

\subsubsection{Using a safe modulus}
\label{sec:safeModulus}
Let $m$ be the modulus used in our scheme. In Section \ref{sec:enc} we specified that $m$ should be a safe prime.
The most important factor in the choice of $m$ is that the cubing operation must be invertible for all (or most) $x \in \mathbb Z_m$.
Moreover, we aim to construct a delay encryption, therefore decryption should be as hard as possible while still being feasible.
To show that safe primes are the optimal choice, we first formalise this requirement.

Let $\mathcal B_m = \{ (\lfloor \log b\rfloor + 1, x) \mid b \equiv 3^{-1} \mmod \mathrm{ord}(x) \land x \in \mathbb Z_m^* \land b < \mathrm{ord}(x) \}$, 
$H_m = \max \{b \mid (b,x) \in \mathcal B_m$ for some $x \}$ and 
$\mathcal H^\epsilon _m = \{  x \mid (b,x) \in \mathcal B_m \,\land \,H_m - b < \epsilon \}$.
The set $\mathcal B_m$ contains the pairs $(b,x)$ such that $b$ is the bitsize of the exponent needed to invert $x^3 \mmod m$.
The set $\mathcal H^\epsilon _m$ represents the subset of $\mathbb Z_m$ for which decryption requires at most $\epsilon$ squarings less than the maximal value $H_m$.
We aim to minimise $K^\epsilon_m = |\ZZ_m^*| - |\mathcal H^\epsilon _m|$ among all $m$ of the same bitsize.

\begin{theorem}\label{th:safe}
If $m$ is a safe prime, then for $3\leq \epsilon < H_m$ $K^\epsilon_m$ is minimal.
\end{theorem}
\begin{proof}
First note that
\begin{equation*}
3^{-1} \mmod \mathrm{ord}(x) = 
\begin{cases}
\frac {2\mathrm{ord}(x) + 1} 3 & \text{if }\mathrm{ord}(x) \equiv 1 \mmod 3 \\
\frac {\mathrm{ord}(x) + 1} 3 & \text{if }\mathrm{ord}(x) \equiv 2 \mmod 3 \\
\text{non-existing} & \text{if }\mathrm{ord}(x) \equiv 0 \mmod 3
\end{cases}
\end{equation*}

For every $m$, $1, -1 \in \ZZ_m^*$. So there are at least two elements where $x^{3} \equiv x \mod m$. Hence, for all $\epsilon < H_m$, $|\mathcal H^\epsilon _m | \leq |\ZZ_m^*| -2$. In particular, $K^\epsilon _m \geq 2$.

If $m=2q+1 > 7$ is a safe prime, $\forall x \in \ZZ_m^*\;\; \mathrm{ord}(x) \in \set{1, 2, 2q, q}$.
Moreover, we have $\gcd(2q, 3) =1$ and $\gcd(q, 3) = 2$.
So, $\inv 3 \mmod 2q = \frac {4q+1} 3$ and $\inv 3 \mmod q = \frac {q+1} 3$. 
The difference in bitsize is at most 2. Since the only elements of order $1$ or $2$ are $1$ and $-1$, $|\mathcal H^3 _m| = |\ZZ_m^*| -2$.
\qed
\end{proof}

The above theorem proves that, using safe primes, we minimise the number of elements whose cube root is easy to compute.
In particular, it shows that computing the cube root requires an exponentiation where the exponent has a very similar bitsize to the prime modulus.
Luckily, any safe prime $p > 7$ also satisfies $\gcd(p-1, 3) = 1$, meaning that cubing is an invertible operation for all $x \in \mathbb Z_p^*$.
We also note that the number of elements with order $q$ is $q-1$ which is equal to the number of elements of order $2q$.
Inverting an element of order $2q$ requires at most 4 more modular multiplications than inverting an element of order $2q$.
Thus, we do not believe there is practically any difference in the delay times.
It follows that when using a padding scheme $\textsf{pad}\colon \mathcal M \to \mathbb Z_p^*$, one only needs to ensure that $p$ does not divide $\textsf{pad}(m)^2 - 1$.
This simple check will ensure that \textsf{pad}$(m)$ does not have order 1 or 2.

\subsubsection{Randomisation in the padding scheme}
\label{sec:randPadding}
When using a deterministic padding scheme, the scheme proposed is not safe as any attacker can ``guess'' the plaintext if the message space is limited.
Consider a scenario where Alice and Bob are trying to flip a coin by sending each other 1 bit and then taking their sum.
Alice and Bob are using delay encryption to ensure neither of them can decide their move after seeing their opponent's choice.
In this scenario, the message space is limited to two options: ``1'' or ``0''.
It follows that Alice must use some random seed in her encryption to stop Bob from guessing what she encrypted and encrypt the result to confirm it.
Since Bob could compute the encryption of ``1'' with many seeds, Alice needs to use a seed long enough so that Bob has a negligible probability of guessing the correct seed  before the delay expires. Note that Bob can perform an expensive precomputation step before receiving the message from Alice and that checking multiple guesses can be performed in parallel.

Say that Alice and Bob are using a prime $p$ of $n$ bits. Alice's seed has length $l$ bits.
To decrypt a message you need at least $n-1$ sequential modular multiplications. 
To encrypt a message you need 2 sequential modular multiplications, where we neglect the cost of the padding scheme used.
Assume that Bob has access to $C$ concurrent computers.
It follows that, in the delay period (i.e.\@ $n-1$ modular multiplications), Bob can perform  $C \frac{n-1} 2$ encryptions.
Therefore, Bob can guess $C \frac {n-1} 2$ random seeds. 
Assuming that the seed is drawn at random from the set of all the $l$-bit strings, then Bob has probability $C \frac {n-1} {2^{l+1}}$ of guessing the correct seed.
Say that Alice wants to bound such a probability by some value $\varepsilon$, then she obtains the following inequality:
\[
l > \log C + \log(n-1) - \log \varepsilon - 1
\]
Picking some values such as $C= 10^{12}$ and $\varepsilon = 10^{-12}$ we obtain $l > 79 + \log (n-1)$.
The above equation represents the bare-minimum requirement for $l$.
However, Bob could perform an extensive precomputation before receiving Alice's message. 
This scenario is somewhat equivalent to Bob trying to guess an $l$-bit key for a symmetric cipher.
As a result, using $l = 256$ should provide enough security for most applications. 
It is important to point out that, in a post-quantum scenario, Grover's algorithm essentially halves the bitsize of search space. 
Therefore, our choice of $l=256$ is expected to guarantee the classical equivalent of 128 bits.
Since our scheme uses primes with thousands of bits, there is practically no upper bound on the number $l$ that Alice can pick.

\section{Chaining the delays}\label{sec:chain}
In the previous section, we showed how to construct a simple delay encryption based on cubing.
As we will discuss in Section \ref{sec:attacks}, the larger the prime $p$ is the easier it is to parallelise the underlying integer multiplication.
It follows that there is an upper bound on the delays one can achieve by cubing.
The purpose of this section is to mitigate this issue by reliably chaining together multiple delays.
The techniques described in this section will also apply to other time-lock puzzles, but we keep our focus on the cubing scheme from Section \ref{sec:tre}.

Let $p$ be a safe prime and define the function $f_p \colon \mathbb Z_p \to \mathbb Z_p$ by $f_p(x) = x^3 \mmod p$.
Our aim is to repeatedly use $f$ to obtain a long delay. 
In particular, we want to find an efficiently invertible function $g_p \colon \mathcal K \times \ZZ_p \to \ZZ_p$ to interleave with $f_p$.
After generating a vector of parameters $\vec{k} = \vector{k_1, \dots, k_n}$, we can construct the function 
$h^{n}_p(x) = (f_p \circ g_p(k_1) \circ \cdots \circ f_p \circ g_p(k_n)) (x)$ as the new delay encryption.\footnote{
In the above description we fixed a safe prime $p$. However, one could use different primes for each use of $f$.
The main disadvantage of this approach is that the sender is required to compute extra safe primes. 
In addition, we do not believe that using different primes gives any security guarantees.
}
The use of $\vec{k}$ for the function $g_p$ is an expedient to allow $g_p$ to receive extra input.
Although the ratio between encryption time and decryption time does not increase, using this construction one can achieve longer delays.

On the theoretical side, $h^n_p$ is a permutation $\mathbb Z_p \to \mathbb Z_p$ and it can be represented using a Lagrange polynomial.
Therefore, it is possible to compute $\inv {(h_p^n)}$ by finding the root of a polynomial in $\mathbb Z_p$. 
As noted in \cite{vdf}, computing a polynomial gcd is sufficient to solve this problem. 
In detail, to invert an injective polynomial $f(x)$ at a point $c$ is enough to compute $\gcd\left(f(x)-c, x^p -x\right)$ since $\inv f(c)$ is the only common root to both polynomials.
Luckily, determining the Lagrange polynomial for $h^n_p$ can be problematic given the size of $p$ if $h^n_p$ is unpredictable enough.
To make $h^n_p$ unpredictable, we must rely entirely on $g_p$ as $f_p$ has a simple algebraic form.
It follows that $g_p$ must be an invertible function which cannot be represented with a simple closed form in $\mathbb Z_p$. 

\subsection{Chaining candidates}
In the setting of verifiable delay functions, a similar problem was studied by Lenstra and Wesolowski \cite{sloth} and by Boneh \etal \cite{vdf}
In particular, Lenstra and Wesolowski \cite{sloth}  prove that $h^n_p$ is secure when $g_p$ is a permutation picked at random and the keys are all equal.
Given the size of $p$, constructing a truly random permutation becomes prohibitive. Therefore, we aim to use appropriate pseudo-random permutations.
It is interesting to note that both papers also suggest simple methods which are likely to be effective in practice, even if there is no theoretical guarantee. 

\subsubsection{Algebraic methods}
For the sake of completeness, we present two algebraic methods used in the context of VDFs, although they lack any security guarantee.
Lenstra and Wesolowski \cite{sloth} proposed the function ``swapping neighbors'' given by:
\[
g_p(x) = 
\begin{cases} 
x+1 & \text{if $x$ is odd} \\
x-1 & \text{if $x$ is even}
\end{cases}
\]
Note that $h^2_p(x) = (x^3 \pm 1)^3 \pm 1 \mmod p = x^9 + 3 x^3 \pm' 3x^6 \pm' 1 \pm 1 \mmod p$ where all $\pm'$ must be the same.
It follows that there is roughly a 25\% chance that $h^2_p(m) = m^9 + 3m^6 + 3 m^3 \mmod p$ for some message $m$. Assume $c$ is the ciphertext (and that $c\neq 0$), then $\gcd(z^9 + 3z^6 + 3z^3 - c, z^{p-1} -1) = z-m$. 
As a result, we believe that the use of the ``swapping  neighbors" function should be avoided.

Boneh \etal \cite{vdf} proposed a new permutation $g_p$ defined over $(\mathbb Z_p)^2$.
In particular, they use the map $g_p(x, y) = (y+c_1, x+c_2)$ for some constants $c_1, c_2$.
Despite its simplicity, there is no apparent security flaw in this construction. 
However, it does require defining $f_p$ over $(\mathbb Z_p)^2$ and clearly $g_p$ is very far from being a secure pseudo-random permutation.

\subsubsection{Format Preserving Encryption} 
A natural candidate for the function $g_p$ is an encryption scheme. 
The field of Format Preserving Encryption focuses on constructing pseudo-random permutations on arbitrary domains.
The FPE schemes approved by the NIST \cite{fpenist} are based on Feistel Networks and aim to encrypt elements of $\Sigma^n$ where $\Sigma$ is an arbitrary alphabet. The constraint $n >1$ rules out the use of $\Sigma = \ZZ_p$, meaning that none of those encryptions can be used on the domain $\ZZ_p$.
However, these schemes can be used to encrypt messages in the domain $\bin^n$ where $n$ is the bitsize of $p$.
Using this in conjunction with the technique of ``cycle walking`` as described  by Black and Rogaway \cite{cycle} gives an encryption $\ZZ_p \to \ZZ_p$.
Given an encryption $E: \{0,1\}^n \to \{0,1\}^n$ and a message space $\mathcal M \subset \bin^n$, the construction of Figure \ref{fig:cyclewalking} results
in an encryption $\mathcal M \to \mathcal M$.
\begin{figure}[h]
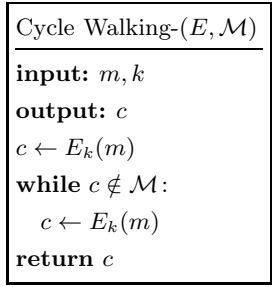

\procb{Cycle Walking-$(E, \mathcal M)$}{
\pcinput m, k\\
\pcoutput c\\
c \gets E_k(m)\\
\pcwhile c \notin \mathcal M\colon\\
\t c \gets E_k(m)\\
\pcreturn c
}
\caption{Cycle walking construction.}
\label{fig:cyclewalking}
\end{figure}

If the encryption $E$ represents a permutation of $\{0,1\}^k$, then this method traverses the cycle of $x$ under $E$ in $\{0,1\}^n$ until it hits an element in $\mathcal M$.
This method is very effective if the ratio $\frac {|\mathcal M|} {2^n}$ is high enough to avoid long encryption chains. 
Moreover, if $E$ is a truly random permutation of $\{0,1\}^n$, then this method will yield a truly random permutation of $\mathcal M$.

We propose a novel efficient FPE method to use in conjunction with the cycle walking construction.
However, we leave its security analysis to a later work.
We write $s[a, b)$ to mean the substring of $s$ starting from the $a$\textsuperscript{th} character (included) and stopping at the $b$\textsuperscript{th} character (excluded).
Figure \ref{fig:bill} describes the algorithm. The underlying idea is to encrypt the input from both ends, each time leaving out a few bits.
\begin{figure}[h]
\procb{Both-Ends Encryption}{
\pcinput s \in \bin^l\\
\pcoutput r \in \bin^l\\
E \gets \text{block cipher with block size } b\\
k_1, k_2 \gets \text{keys for } E\\
x \gets l \mmod b\\
t \gets E_{k_1}\left(s[0, b)\right) \concat E_{k_1}\left(s[b, 2b)\right) \concat \dots \concat E_{k_1}\left(s[l-x-b, l-x)\right) \concat s[x, l)\\
r \gets t[0, x) \concat E_{k_2}\left( t[x, x+b) \right) \concat E_{k_2}\left(t[x+b, x+2b)\right) \concat \dots \concat E_{k_2}\left(t[l-b, l)\right)\\
\pcreturn r 
}
\caption{FPE encrypting from both ends.}
\label{fig:bill}
\end{figure}
The FPE algorithms described here are likely to be the most efficient choice chaining candidate described in this section.
In particular, if we pick the prime $p$ to have a bitsize close to a multiple of the block size of a good block cipher, then we can use the cycle walking technique without the need of any other FPE scheme.
Assuming the use of a 128-bit-block-sized cipher such as AES, we looked at all the safe primes with a bitsize in the range $[30 000, 100 000]$ on the database \url{prime.utm.edu}.
The only primes that we consider practical are:
\begin{enumerate}
\item $1030710193 \times 2^{44001} + 3$ bitsize: $44031 = 344\cdot 128 -1$
\item $1022253375 \times 2^{43489} - 1$ bitsize: $43519 =  340 \cdot 128 -1$
\item $168851511 \times 2^{33251} - 1$ bitsize: $33279 = 260 \cdot 128 -1$
\end{enumerate}
Using any of the above primes, roughly 2 or 3 attempts will be required on average.
As a result, the number of encryptions varies from 520 to 1032 on average.
This is drastically less than the number of encryption needed when using Both-Ends Encryption and cycle walking (on average 2188 if $p$ has 70034 bits).
Using the FF1 method approved by NIST \cite{fpenist} and cycle walking requires on average 5520 encryptions.
Given the highly optimised AES instruction set that recent CPU support, these numbers of encryptions can be carried out relatively efficiently.
To better understand the relative speed of this chaining method, we used the open-source library LibreSSL \cite{libressl} to benchmark AES.
Using the same setup that delays 70034-bit numbers in 0.8ms, AES achieves a speed of roughly 4GB per second. 
Hence, 5520 encryptions can be performed in roughly 0.02ms. 
Table \ref{tab:chainspeed} shows that the speed of the various chaining methods and proves that the FPE approach is the most efficient.

\subsubsection{Card Shuffling}
Card shuffling schemes are essentially families of algorithms to construct ciphers where the domain can be chosen at will.
Assume the message space is $\{0, \dots, N-1\}$, then message $i$ is represented by the $i$\textsuperscript{th} card in a deck of $N$ cards.
After shuffling the deck, that card will end up in position $j$ for some $j\in \{0, \dots, N-1\}$, therefore the encryption of $i$ is $j$.
An oblivious shuffle is a shuffle where the trajectory of a card can be computed without calculating the paths of the other cards.
In the literature, we could find only a few FPE schemes based on oblivious card shuffling \cite{son,srs,thorp,mac}.
We now proceed to present them.

The earliest card shuffle present in the literature is the Thorp shuffle described in Figure \ref{fig:Thorp}.
Following the description by Morris \etal \cite{thorp} the shuffle is performed as follows: cut the deck in half and pick up each half; according to a coin flip, drop the last card from the left deck, then the right deck or vice versa.
Figure \ref{fig:Thorp} shows the algorithm for the shuffle where only the trace of the input is computed.
In order to reverse the shuffle, the sampling of the coin flip (variable $b$ in Figure \ref{fig:Thorp}) must be deterministic.
As with the other card shuffles, in order to guarantee high security, this cipher will need $R \in O(\log p)$. 
Given our large sizes of $p$, this approach is quite expensive.
\begin{figure}[h]
\procb{Thorp}{
\pcinput N, x \in \set{0, \dots, N-1}\\
\pcoutput y \text{ the new position of the card } x\\
\pcrepeat{R}\colon\\
\t y \gets x + \frac N 2 \mod N\\
\t b \sample \bin\\
\t \pcif x < y\colon\\
\t\t x \gets 2x+1-b\\
\t\pcelse \colon \\
\t\t x \gets 2x+b\\
\pcreturn x
}
\caption{Thorp shuffle.}
\label{fig:Thorp}
\end{figure}

The swap-or-not shuffle by Hoang \etal \cite{son} follows a similar structure: pair cards $x$ and $K-x\mmod N$ and swap the cards according to a coin flip.
Figure \ref{fig:SoN} shows the algorithm for this shuffle and, once again, the coin flip must be deterministic.
Hoang \etal\  suggest the use of round functions $F_i$ that should be applied to $\max(x, y)$.
As a result, the sequence of coin flips cannot be precomputed since each flip depends on the current state of the algorithm.
\begin{figure}[h]
\procb{Swap-Or-Not}{
\pcinput N, x \in \set{0, \dots, N-1}\\
\pcoutput y \text{ the new position of  the card } x\\
\pcrepeat {R} \colon\\
\t K \sample \set{0, \dots, N-1}\\
\t y \gets K-x \mmod N\\
\t b \sample \bin\\
\t \pcif b=1 \colon\\
\t[2] x \gets y\\
\pcreturn x
}
\caption{swap-or-not shuffle.}
\label{fig:SoN}
\end{figure}

Ristenpart and Yilek \cite{mac} designed the  mix-and-cut shuffle as a way of strengthening other shuffles. 
In a later paper by Morris and Rogaway \cite{srs}, the mix-and-cut shuffle with swap-or-not is slightly modified to improve efficiency without impacting the security guarantees.
Figure \ref{fig:MaC} shows the algorithm for this improved version.
On an abstract level, the construction can be described as follows: a ``standard" shuffling is used, then the deck is cut in half and recursively only the first half is shuffled.
\begin{figure}[h]
\procb{Mix-And-Cut}{
\pcinput N, x \in \set{0, \dots, N-1}\\
\pcoutput y \text{ the new position of the card } x\\
\pcif N =1\colon \pcreturn 0\\
x \gets \text{Swap-Or-Not}(N, x)\\
\pcif x < \frac N 2 \colon\\
\t \pcreturn \text{Mix-And-Cut}\left(\frac N 2, x\right)\\
\pcreturn x
}
\caption{mix-and-cut shuffle.}
\label{fig:MaC}
\end{figure}
As mentioned earlier, the use of card shuffling ciphers is relatively inefficient as the number of iterations for each shuffle is  $O(\log p)$.
In particular, we implemented the three shuffles mentioned above using the GMP library \cite{GMP}. 
By setting $N$ to be our usual prime of 70034 bits and $R = 70000$, we tested these shuffling algorithms and, discounting the cost of sampling the random keys and coin flips, we obtained the timings  shown in Table \ref{tab:chainspeed}. 
This shows that the computation of these shuffles would dominate the overall complexity of the chained delay $h^n_p$.
On the other hand, these ciphers are the candidates which more closely represent a random permutation.
Therefore, applications of the chaining techniques that require very high security should use one of the card shuffling ciphers described here.
It is important to note that in our use case, there is no need for the shuffling to withstand cryptanalysis. 
As a result, we believe that using a shuffle with only few iterations (small $R$) it is very likely to provide enough security for our needs.

\begin{table}[h]
\centering
\makebox[\textwidth]{
\begin{tabular}{|c  |c |c |c| c|}
\hline
Cubing & AES-256 & Thorp & Swap-Or-Not & Mix-And-Cut\\
\hline
0.8ms & 0.02ms & 37ms & 47ms & 83ms\\
\hline
\end{tabular}}
\caption{Comparison of chaining methods' speed.}
\label{tab:chainspeed}
\end{table}

\subsection{Conclusion}
Using the chaining construction, one can extend the delays constructed via cubing.
Long delays achieved via long chains are still impractical since this chaining construction does not increase the ratio between encryption time and decryption time.

Among the few candidates we presented, we believe that using FPE methods with standard block ciphers is the most balanced option.
They are efficient methods and provide enough security for our particular use in chains.
We do not believe it is necessary to change keys between the applications of $g_p$ in the same chain. 
However, we discourage using the same keys for multiple chains.

Card shuffling schemes should be used when the highest degree of security is needed.
In this context, we suggest using randomly generated keys. 
Moreover, the bit sampling present in the shuffles should depend on the entire internal state of the algorithm. 
This is likely to result in the most unpredictable permutations, since neither the keys nor the coin flips can be computed before they are required.

Despite their frequent use, the algebraic methods represent the least secure option. 
As noted at the start of Section \ref{sec:chain}, the ability to express the whole chain as a polynomial results in a simple attack.
Therefore, we suggest using other alternatives.

\section{Breaking the delay}
\label{sec:attacks}
The only security requirement for a delay encryption scheme is that nobody should be able to decrypt the delayed message before its intended delay.
Since chaining can be done securely, in this section we analyse possible attacks to the cubing operation.
In particular, we aim at estimating the speedup factor that a resourceful entity can gain against an average user.
We focus on four different attack vectors: exponentiation, modular reduction, integer multiplication and hardware performance.

\subsection{Repeated squarings}
Since the introduction of the RSA cryptosystem, researchers have looked at ways to optimise the modular exponentiation process.
Let $a, b$ and $m$ be integers with $b$ having $n$ bits. 
All the optimisations present in the literature aims at reducing the number of modular multiplications needed to compute $a^b \mod m$.
However, there seems to be an implicit lower bound of $n-1$. 
As an example, we describe the state-of-the-art ``sliding window" algorithm.
We write $b[i, j)$ for the integer $\floor{ (b  - \floor {b / 2^{n-i}} 2^{n-i}) / 2^{n-j}}$ (i.e.\@ the number composed by $b$'s bits in positions $\set{i, i+1, \dots, j-1}$).
Assuming $n = Nw$, $b = \sum_{i=0}^{N-1} b[iw, iw+w) 2^{iw}$, so $a^b = \prod_{i=0}^{N-1} a^{2^{iw}b[iw, iw+w)}$. 
Figure \ref{fig:exp} shows the pseudocode for this algorithm.
\begin{figure}[h]
\procb{Sliding Window Exponentiation}{
\pcinput a, b, w\\
\pcoutput a^b\\
\text{construct table } T \text{ of size } 2^w\\
\pcforeach x \in \bin^w\colon\\
\t T[x] = a^x\\
c\gets T[b[0, w)]\\
i \gets w\\
\pcrepeat {\frac n w -1}\colon\\
\t \pcrepeat {w}\colon\\
\t[2] c \gets c^2\\
\t c \gets c\cdot T[b[i, i+w)]\\
\t i \gets i+w\\
\pcreturn c
}
\caption{Sliding window exponentiation assuming $n = 0 \mod w$.}
\label{fig:exp}
\end{figure}
The GMP library \cite{GMP} implements a slightly more optimised version of this algorithm.
Similar to all the other methods present in the literature, this technique aims at reducing the number of multiplications required apart from the repeated squaring. As we will explain below, most of these extra multiplications can be parallelised, therefore these nuances are less important since we are interested only in the sequential cost of modular exponentiation.
In this regard, there is an implicit lower bound of $n-1$ multiplications.
By parallelising the above algorithm, one can get extremely close to this lower bound.
Write $b = \sum_{i=0}^{N-1} b_i 2^{iw}$, i.e.\@ $b_i = b[i, i+w)$.
For $i \neq j$ $a^{b_i 2^{iw}} $ and $a^{b_j 2^{jw}} $ can be computed independently and multiplied together as they become available.
In a perfect world, when the computation of $a^{b_{N-1}2^{(N-1)w}}$ is completed, all other values will have been multiplied together, so that only one extra multiplication is needed.
As a result, on average, the number of sequential  multiplications needed is $n+ \frac {w-1} 2$.
In Figure \ref{fig:timeline}, this idea is represented graphically.

\begin{figure}[h]
\centering
\begin{tikzpicture}[scale=0.55]
	\begin{pgfonlayer}{nodelayer}
	    \node [style=none] (16) at (-9, 4.1) {$a^{b_0}$};
		\node [style=none] (0) at (-8, 4) {};
	    \node [style=none] (17) at (-9, 2.1) {$a^{b_12^w}$};
		\node [style=none] (1) at (-8, 2) {};
		\node [style=none] (18) at (-9, 0.1) {$a^{b_2 4^w}$};
		\node [style=none] (2) at (-8, 0) {};
		\node [style=none] (3) at (-8, -2) {$\vdots$};
		\node [style=none] (19) at (-10, -3.9) {$a^{b_{N-1} 2^{(N-1)w }}$};
		\node [style=none] (4) at (-8, -4) {};
		\node [style=none] (5) at (-6.75, 4) {};
		\node [style=none] (6) at (-4.75, 2) {};
		\node [style=none] (20) at (-2, 3.2) {$a^{b_0 + b_12^w}$};
		\node [style=none] (7) at (-3, 0) {};
		\node [style=none] (8) at (3, -4) {};
		\node [style=none] (9) at (3, -2) {};
		\node [style=none] (10) at (-3.75, 3) {};
		\node [style=none] (11) at (-2, 1) {};
		\node [style=none] (12) at (4.5, -3) {};
		\node [style=none] (20) at (5, -2.8) {$a^b$};
		\node [style=none] (13) at (-0.25, -0.25) {};
 		\node [style=none] (14) at (1, -2) {};
		
	\end{pgfonlayer}
	\begin{pgfonlayer}{edgelayer}
		\draw [thick] (0.center) to (5.center);
		\draw [thick](1.center) to (6.center);
		\draw [thick](5.center) to (10.center);
		\draw [thick](6.center) to (10.center);
		\draw [thick](2.center) to (7.center);
		\draw [thick](7.center) to (11.center);
		\draw [thick](10.center) to (11.center);
		\draw [thick](4.center) to (8.center);
		\draw [thick](8.center) to (12.center);
		\draw [thick](12.center) to (9.center);
		\draw [thick, dashed](11.center) to (13.center);
		\draw [thick, dashed](9.center) to (14.center);
	\end{pgfonlayer}
\end{tikzpicture}
\caption{Parallelisation of the decryption process.}
\label{fig:timeline}
\end{figure}
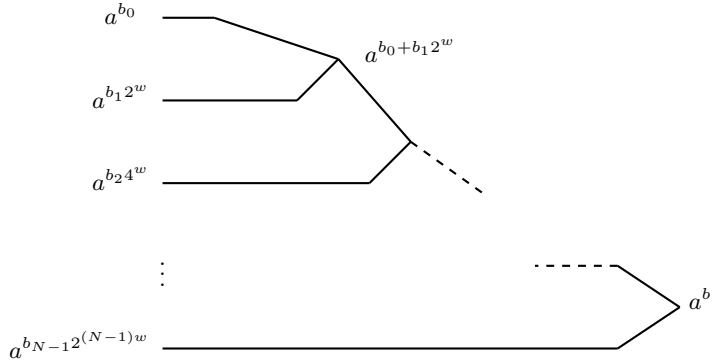

Essentially, the above approach reduces the exponentiation process to computing only the repeated squaring. 
This shows that it is unreasonable to assume that the exponentiation $a^b$ requires more than $\lfloor \log b \rfloor$ sequential multiplications. 
On the other hand, the lack of algorithms that can outperform the process of repeated squaring is a favourable sign for our Conjecture \ref{seqAssump}.
Under the assumption that repeated squaring cannot be parallelised, the only attack vector left is improving each singular modular multiplication.
There are mainly two approaches to compute a modular multiplication $ab \mmod c$.  The naive approach computes the multiplication $ab$ and then reduces it  modulo $c$.
The direct approach computes $ab \mmod c$ in a single algorithm by keeping most of the intermediate values below $c$.
This latter approach can be much more memory efficient than the former. However only a few such algorithms can be found in the literature (e.g.\ \cite{Chen}).
On the other hand, the naive approach seems to be the most widely used as one can take advantage of the extensive research on both integer multiplication and modular reduction.
Since memory restrictions are not a concern in our setting, we will focus on the fastest methods for squaring and modular reduction.

\subsection{Modular reduction}
In this section we look at the algorithms for computing modular reduction and how to use those during modular exponentiation.
Our aim is to find the best approach within the setting of repeated squaring.
Let $b$ be an integer of $n$ bits and $a$ an integer of $m$ bits with $m \geq n$. We aim to compute $a \mmod b$.
\footnote{In the literature one often finds the extra constraint that $m \leq 2 n$. 
In this section we will try to avoid such assumption, and we will clearly state when this is not the case.}

We have discovered only three different approaches to modular reduction in the literature.
The most common are Montgomery's algorithm \cite{Montgomery} and Barret's reduction \cite{Barret}. The third approach is based on table lookup.

Barret's algorithm is based on the equality $a \mod b = a - \floor{ \frac a b} b$. 
\begin{figure}[h]
\begin{minipage}{0.47\textwidth}
\procb{Barret}{
\textbf{input: } a, b\\
\textbf{output: } a \mmod b\\
q \gets \floor{\frac {2^{m+1}}{b}} (\text{can be precomputed})\\
r \gets a - \left(\left( a \cdot q \right) \gg \left( m+1 \right) \right) \cdot b\\
\pcif r > b\colon\\
\t r \gets r-b\\
\pcreturn r
}
\caption{Barret's reduction algorithm.}
\label{fig:barret}
\end{minipage}
\begin{minipage}{0.47\textwidth}
\procb{REDC}{
\textbf{input: } a, b, R\\
\textbf{output: } a\inv{R} \mmod b\\
b' \gets \inv{b} \mmod R\;\; (\text{can be precomputed})\\
m \gets \left(a \mmod R\right) \cdot b' \mmod R\\
t \gets \left( a + m \cdot b\right) / R\\
\pcif t > b\colon\\
\t t \gets t-b\\
\pcreturn t
}
\caption{Montgomery's reduction algorithm.}
\label{fig:mont}
\end{minipage}
\end{figure}

The description in Figure \ref{fig:barret} differs from the original version described by Barret \cite{Barret} as we use $q \coloneqq \left\lfloor \frac {2^{m+1}} b \right\rfloor$ rather than 
$q \coloneqq \left\lfloor \frac {2^{m}} b \right\rfloor$. This allows us to perform only one subtraction by $b$ rather than two. 
The extra bit in $q$ is unlikely to cause any performance hit.
Therefore, the overall time complexity is: 2 multiplications and 2 subtractions. All of these involve $m(+1)$ bit numbers, apart from the last subtraction.
Note that computing $q$ may require a considerable amount of precomputation, but in our encryption scheme this needs to be done only once per ``public key''.

Figure \ref{fig:mont} shows the pseudocode for Montgomery's reduction.
A key aspect of this algorithm is that, given $a$ and $b$, it computes $aR^{-1} \mod b$ for some fixed $R$.
In order to use the algorithm of Figure \ref{fig:mont}, the modular exponentiation needs to be carried out in ``Montgomery form''.
That is, to compute $c^2 \mmod b$, the base $c$ needs to be converted into $c_M \coloneqq cR \mmod b$. 
Therefore REDC$(c_M^2) = c^2R \mmod b$. 
The final result will also need to be converted back from the ``Montgomery form'', which can be done by using {REDC} once more.
In order for {REDC} to work correctly, one needs $a < Rb$ and $(R, b) = 1$. 
Moreover, $R$ should be a power of 2, otherwise the algorithm is not worthwhile using.
It follows that the overall time complexity is: 2 multiplications, 1 addition and 1 subtraction.
Comparing this with Barret's reduction, we see that the overall complexity is very similar.

When considering the lookup table approaches, the most recent results we could find are by Cao \etal\cite{Cao} and by Will and Ko \cite{Cao}.
Their approach is relatively simple: one precomputes the map $x \mapsto x \mmod b$ for some specific values of $x$ and then represents $a$ with a suitable vector of $x$'s.
This leads to an addition-based algorithm. Since the number of additions is linear, these approaches have a worse theoretical time complexity.
It follows that we do not believe these approaches can consistently outperform the previous two methods described.
As a result, throughout this paper we will assume that the complexity of the modular reduction consists of two multiplications of $\lfloor \log a\rfloor$ bits and two additions of $\lfloor \log a \rfloor$ bits.

Finally, we need to analyse how to best use modular reduction within the exponentiation process.
Let $p$ be an $n$-bit prime and let $x, y \leq p$ be integers. Consider the process of computing $x^y \mod p$.
If the computation is done in Montgomery form, then the reduction should be interleaved with the multiplication to prevent the factor $R$ from growing.
Otherwise, one may choose to take the reduction modulo $p$ every $m$ steps for any $m\geq 1$.
Let $M(n)$ be the time required to multiply two $n$-bit numbers and let $R(n)$ be the time required to reduce an $n$-bit number modulo $p$.
Interleaving multiplication and reduction leads to a time complexity of $k(M(n) + R(2n))$ for some $k$ which represents the number of modular multiplications needed.
If the reduction is performed only every $m$ multiplications, then the time complexity becomes
\[
\frac k m \left(\sum_{i=0}^{m-1} M(2^i n) + R(2^m n)\right)
\]
Recall that $R(n) \in O(M(n))$, so depending on $M(\cdot)$ the optimal choice for $m$ could be either 1 or $k$.
If multiplication can be done in sequential time $O(\log n)$ (see Section \ref{sec:circ}), then one is better off picking $m=k$.
If $M(n)\geq \alpha n$ for some fixed $\alpha$, then the optimal value for $m$ is 1.

\subsection{Integer multiplication}
Although repeated squaring is a sequential process, integer multiplication is definitely not. 
This makes the study of integer multiplication crucial for understanding where a malicious entity could gain an unfair advantage.

In this section, we analyse the main known integer multiplication algorithms in order to find the one most suited to our application.
Let $a, b$ be two $n$-bit numbers, we aim to compute $ab$.

The first family of multiplication algorithms that we must mention are the so-called Toom-Cook algorithms \cite{Bodrato:WAIFI2007}.
In that paper, the famous Karatsuba-Ofman multiplication is nothing but Toom-2.
\begin{figure}[h]
\procb{Toom-$k$}
{
\pcinput a, b \text{ both $n$ bits long}\\
\pcoutput ab\\
\lambda \gets 2^{\floor{\frac n k }}\\
\text{compute $a_i$'s such that } a = \sum_{i=0}^{k-1} a_i \lambda^i\\
\text{compute $b_i$'s such that } b = \sum_{i=0}^{k-1} b_i \lambda^i\\
A(x) \coloneqq \sum_{i=0}^{k-1} a_i x^i\\
B(x) \coloneqq \sum_{i=0}^{k-1} b_i x^i\\
\text{pick a set } \mathcal P = \set{p_1, \dots, p_{2k-1}}\\
\mathcal V \gets \set{ v_i \mid v_i = A\left(p_i\right) B\left(p_i\right)\;\; \forall\, p_i \in \mathcal P }\\
P(x) \gets \text{ polynomial such that } P(p_i) = v_i \text{ for } 1\leq i \leq 2k-1\\
\pcreturn P(\lambda)
}
\caption{Toom-Cook algorithm.}
\label{fig:toom}
\end{figure}
The asymptotic complexity of the Toom-$k$ algorithm is $O\left(n^{\frac {\log(2k-1)}{\log k}}\right)$.

Figure \ref{fig:toom} describes this multiplication algorithm.
The choice of $\mathcal P$ has a big impact on the algorithm performance, therefore the values $p_i$ are always chosen so that computing $A(p_i)$ and $B(p_i)$ is efficient.
Moreover, the interpolation to obtain $P(x)$ is often hardcoded in order to minimise the number of operations needed.
This algorithm can be parallelised by computing each $v_i$ on concurrent threads.
We attempted modifying GMP's implementation of this algorithm by parallelising it with low-level software, 
but this was not sufficient to achieve any performance improvements when multiplying numbers of hundreds of thousands of bits or smaller.

The choice of the parameter $k$ is crucial and depends both on $n$ and the hardware on which the algorithm is run.
As an example, the GMP library lets the user optimise the thresholds for picking $k$. 
During our testing on numbers with $70034$ bits, we found that the algorithm with $k=8$ gave the best performance.

Another important family of multiplication algorithms is based on Fourier Transforms.
The idea behind these algorithms is similar to the principle behind the Toom-Cook algorithms.
The multiplicands $a, b$ are interpreted as polynomials $A(x), B(x)$. 
The Fast Fourier Transform (FFT) is used to evaluate these polynomials at the root of unity of some field. 
Recursively, the algorithm computes the multiplication of these points, and then interpolates them to obtain the polynomial $P(x) = A(x)B(x)$.

One of the most famous FFT algorithms is the Sch\"onhage-Strassen algorithm \cite{SSGMP} which achieves a time complexity of $O(n \log n \log \log n)$.
This algorithm is very appealing as it works in the polynomial ring $(\mathbb Z/\langle2^k +1\rangle)[x] / (x^T +1)$ where both $k$ and $T$ are powers of $2$.
This allows the algorithm to avoid floating-point calculations and their related precision issues. 
This is likely the most popular FFT-based integer multiplication algorithm, and it is the one implemented within the GMP library.

Another FFT algorithm worth mentioning is the one proposed by Harvey and van der Hoeven \cite{mult}  since its time complexity is $O(n \log n)$.
This algorithm works in a multidimensional polynomial ring over the complex numbers, so any implementation must be careful to handle floating-point numbers with enough precision.
We are not aware of any implementation of this algorithm, and we expect that it will outperform the Sch\"onhage-Strassen algorithm only for very large $n$.

Since the FFT algorithms follow the same structure as the Toom-Cook family, they could also be parallelised relatively easily.
As we will see in the next section, these algorithms are the most common option when trying to parallelise integer multiplication on a GPU.

We would like to point out that the existence of an optimised open-source library, such as GMP, is vital in limiting the performance advantage a malicious entity could obtain via software improvements.
Based on our own experiments, the GMP library is of very high quality in respect of optimisation and coverage.

\subsection{Faster Hardware}
In this section, we discuss how using different hardware might lead to performance improvements.
As far as CPUs are concerned, the difference in single core performance does not appear to be very wide.
In particular, a smartphone's CPU has a clock speed of around 2GHz, while a top tier desktop processor can only reach up to 5.3GHz.
It follows that the gap in clock speed can be bounded by 3.
Differences in the CPU architecture can lead to different performances even with equal clock speed.
To estimate this factor, we used a table \cite{GMPtable} published on the GMP library website. 
This table lists the average cost in terms of clock cycles for many low-level tasks on different CPU architectures.
For instance, summing two $64n$-bits numbers on an AMD Zen 3 processor takes $n$ cycles, while doing the same task on an ARM A5 neon takes $8.66n$ cycles. 
Similarly, using an Intel Atom architecture to square a number with the trivial algorithm takes roughly 9.7 cycles per 32 bits, while performing the same task on an Apple M1 chip takes 1.41 cycles per 64 bits.
These figures show that the gap between old architectures and new ones can be much higher than the gap in clock speed.
As a result, we suggest the cautious approach of assuming that malicious entities could obtain a 15 times speed-up compared to mid-range CPUs.

Another approach that could lead to faster decryption is to perform the repeated squaring on different hardware.
The suggestion of Field Programmable Gate Arrays (FPGA) often appears in the literature.
The use of this technology seems to allow for theoretical improvements, but we do not believe that commercially available FPGAs could have such impact.
As an example, Woo \etal \cite{FPGA1} implements three different multiplication algorithms on FPGAs and  these algorithms can multiply two 8-bit numbers with a latency of 5 cycles. Similarly, Kakacak \etal \cite{FPGA2} claim to achieve better timings than the state of the art in the literature, yet their implementations require 7ns to multiply two 64-bits numbers.
A recent architecture like AMD Zen 2 using GMP requires around 1.75 clock cycles per 64-bits when multiplying numbers. 
Therefore, using a processor like the AMD Ryzen 7 3800X, which has a base clock speed of 3.9GHz, one can multiply two 64-bits numbers in roughly 0.4ns.
It follows that the results in the literature are far from being competitive with more traditional  hardware. 
In addition to the performance gap, the AMD processor mentioned above  costs around £300, which is likely to be much less than the cost of a custom FPGA.
Therefore, we do not believe that, at the time of writing, FPGAs are a viable way to achieve any performance boost.

Another common type of hardware used in the literature are GPUs.
The literature reveals many attempts to parallelise multiplication algorithms on many-core GPUs. 
However, we could not find an exhaustive comparison of different GPU and CPU architectures.
Chang \etal \cite{GPU} used a GTX1070 to multiply two numbers of 192K-bits in 1.12ms. Using GMP and an Intel i7-6920HQ we perform the same multiplication in roughly 5ms.
Emeliyanenko \cite{GPU10x} used a GeForce GTX280 to achieve roughly a 10 times speed-up compared to a quad-core Intel Xeon E5420 on ``moderate'' size multiplicands (around tens of thousands of bits). 
On the other hand, Ochoa-Jiménez \etal \cite{CPUvGPU} implement the RSA scheme on both CPU and GPU. 
After testing their implementations using key sizes up to 3072 bits, they conclude that ``in spite of its massive parallelism, we observe that GPU implementations of RSA are considerably slower than their CPU counterparts''. 

Using only the available results, we suggest the cautious approach of assuming a 20 times speed-up when using a GPU compared to mid-range CPUs.

\subsection{Custom Circuits}
\label{sec:circ}
In this section we analyse whether some party could improve the decryption time by building a custom circuit.
In particular, Theorem \ref{th:circpow} below proves that one can obtain a great speed-up compared to the previous methods.

Here, we use the standard circuit model where we bound the input on any logic gate by two bits.
In this model, one can prove a few simple propositions that we list below.
\begin{proposition}
Let $x, y$ be two $n$-bit numbers. There is a circuit of depth $O(\log n)$ and size $O(n^2)$ that computes $x+y$.
\end{proposition}
\begin{proposition}
Let $x_1, \dots, x_k$ be $k$ $n$-bit numbers. There is a circuit  of depth $O(\log n + \log k)$ and size $O(nk + n^2)$ that computes $\sum_{i=1}^k x_i$.
\end{proposition}
\begin{proposition}
Let $x, y$ be two $n$-bit numbers. There is a circuit of depth $O(\log n)$ and size $O(n^2)$ that computes $xy$.
\end{proposition}

For the next two results, one needs to perform some polynomial time precomputation and hardwire the result in the circuit.

\begin{theorem}[Beame \textit{et al.}\ \cite{Beame}]\label{th:itsq}
Let $x_1, \dots, x_k$ be $k$ $n$-bit numbers. There is a circuit of depth $O(\log n + \log k)$ that computes $\prod_{i=1}^k x_i$.
\end{theorem}
\begin{theorem}[Hamano \textit{et al.}\ \cite{Hamano}]\label{th:circpow}
Let $x, y, m$ be three numbers of $n$ bits. For any $\alpha > \frac 1 {\log n}$, there is a circuit of depth $O(n\frac {\alpha+1}\alpha)$ and size $O(n^{3+ 2\alpha}/(\alpha \log n))$ that computes $x^y \mmod m$.
\end{theorem}

Hamano \etal\cite{Hamano} construct a circuit for Theorem \ref{th:circpow} using the high-level structure of Figure \ref{fig:Hamano}.
\begin{figure}[h]
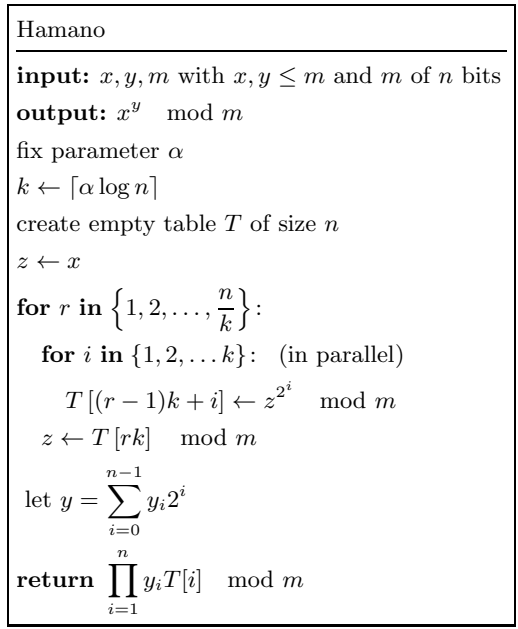

\procb{Hamano}
{
\pcinput x, y, m \text { with } x,y \leq m \text{ and } m \text{ of } n \text{ bits}\\
\pcoutput x^y \mod m\\
\text{fix parameter } \alpha\\
k \gets \ceil{\alpha \log n}\\
\text{create empty table } T \text{ of size } n\\
z \gets x\\
\pcfor r \pcin \; \set{1, 2, \dots, \frac n k}\colon\\
\t \pcfor i \pcin \; \set{1, 2, \dots k}\colon \;\;(\text{in parallel})\\
\t \t T\left[(r-1)k + i\right] \gets z^{2^i} \mod m\\
\t z \gets T\left[rk \right] \mod m\\
\text{ let } y = \sum_{i=0}^{n-1} y_i 2^i\\
\pcreturn \prod_{i=1}^n y_i T[i] \mod m
}
\caption{Hamano \textit{et al.}'s modular exponentiation circuit.}
\label{fig:Hamano}
\end{figure}
The last multiplication can be computed in parallel by arranging the multiplications in a binary tree.
Therefore, the last step can be performed by a circuit of depth $O(\log^2 n)$. 
All the modular reductions are performed using Barret's reduction, and so they can be computed by a circuit of depth $O(\log n)$.
The key step of this algorithm is the computation of $z^{2^k}\mmod m$ in depth $O(k)$.
Hamano \textit{et al.}\ suggest using the circuit constructed by Okabe \etal \cite{Okabe} (see Figure \ref{fig:Okabe}).
\begin{figure}[h]
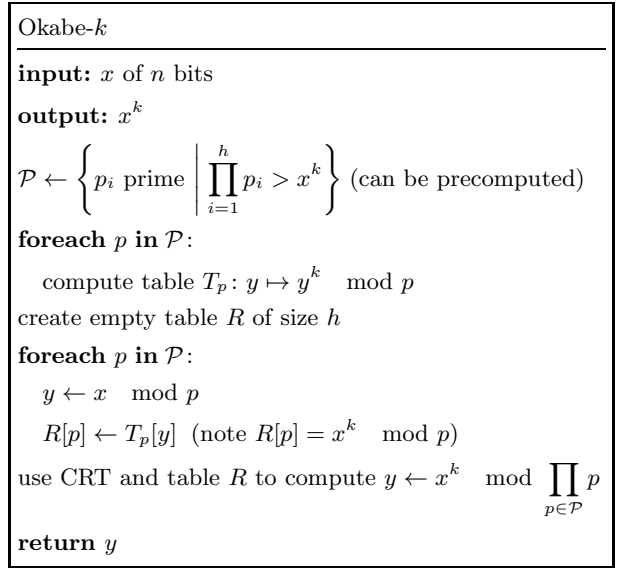

\procb{Okabe-$k$}{
\pcinput x \text{ of } n \text{ bits }\\
\pcoutput x^k\\
\mathcal P \gets \set{p_i \text{ prime } \middle| \; \prod_{i=1}^h p_i > x^k}\; (\text{can be precomputed})\\
\pcforeach p \pcin \; \mathcal P\colon\\
\t \text{compute table } T_p\colon y \mapsto y^k \mod p\\
\text{create empty table } R \text{ of size } h\\
\pcforeach p \pcin \; \mathcal P\colon\\
\t y \gets x \mod p\\
\t R[p] \gets T_p[y] \;\;(\text{note } R[p]= x^k \mod p)\\
\text{use CRT and table $R$ to compute } y \gets  x^k \mod \prod_{p \in \mathcal P} p\\
\pcreturn y
}
\caption{Okabe \textit{et al.}'s circuit for exponentiation.}
\label{fig:Okabe}
\end{figure}
The key to its performance is that the primes $p_i$ are relatively small and so computations modulo $p_i$ can be performed efficiently.
In particular, the precomputed tables can be calculated in time polynomial in $n$. 
However, within the scope of Figure \ref{fig:Hamano} one would need to precompute tables for $y\mapsto y^l \mod p$ for all $p$ and also for all $l\in \{2, 2^2, \dots, 2^k\}$. This makes the memory requirements for this approach prohibitive in our use case.

Beame \etal\cite{Beame} construct a circuit for Theorem \ref{th:itsq} in a more complex but memory efficient manner than Okabe \textit{et al.}\ (see Figure \ref{fig:Beame}).
\begin{figure}[h]
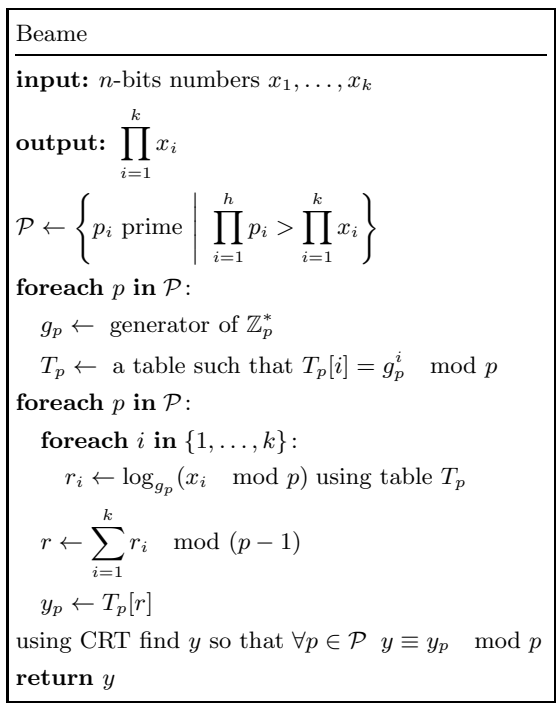

\procb{Beame}{
\pcinput n\text{-bits numbers } x_1, \dots, x_k\\
\pcoutput \prod_{i=1}^k x_i\\
\mathcal P \gets \set{p_i \text{ prime } \middle|\;\; \prod_{i=1}^h p_i > \prod_{i=1}^k x_i}\\
\pcforeach p \pcin\; \mathcal P\colon\\
\t g_p \gets \text{ generator of } \ZZ_{p}^*\\
\t T_p \gets \text{ a table such that } T_p[i] = g_p^i \mod p\\
\pcforeach p \pcin \; \mathcal P\colon\\
\t \pcforeach 	i \pcin \; \set{1, \dots, k}\colon\\
\t[2] r_i \gets \log_{g_p} (x_i \mod p) \; \text{using table }T_p\\
\t r \gets \sum_{i=1}^k r_i \mod (p-1)\\
\t y_p \gets T_p[r]\\
\text{using CRT find } y \text{ so that } \forall p \in \mathcal P\;\; y \equiv y_p \mod p\\
\pcreturn y
}
\caption{Beame \textit{et al.}'s iterated multiplication circuit.}
\label{fig:Beame}
\end{figure}
The approach by Beame \textit{et al.}\ can easily be modified to compute $x^k$ rather than $\prod_{i=1}^k x_i$.
Compared to  Figure \ref{fig:Okabe} this approach is more complex and likely slower, but it has the advantage that the precomputed tables do not change if $k$ varies. 
It follows that this approach is much more memory efficient when used in Figure \ref{fig:Hamano}.

Another source of memory consumption of the algorithms from Figures \ref{fig:Okabe} and \ref{fig:Beame} is that the reduction modulo a prime $p$ is done using lookup tables.
That is, to compute $x \mod p$ the algorithms use a table of the values $2^i \mmod p$ for $i\in \{1, \dots, n\}$ where $n$ is the size in bits of $x$.

There are  two main issues with the practical realisation of the circuit proposed by Hamano \textit{et al.}: the size of the circuit and the memory requirements for the precomputation.
For instance, picking $\alpha = 1$ and $n=70034$ leads to a circuit of size (abusing notation) $O(10^{23})$.
According to a recent news article \cite{largePC} the largest computer in 2020 has a transistor count of only $2.6 \times 10^{12}$ indicating that constructing a circuit of $10^{23}$ gates is infeasible.
By lowering the constant $\alpha$, the size of the circuit can be reduced.
Choosing $\alpha = 0.062 \approx \frac 1 {\log n} $ we obtain a circuit of size $O(10^{15})$ which is still above what we believe possible.
Moreover, using such $\alpha$ the circuit will have depth $O(n \log n)$.

Table \ref{tab:memcirc} shows the memory requirement for a few choices of $\alpha$ when $n$ is fixed to 70034.
As one can see, the memory requirements become prohibitive very quickly as $\alpha$ approaches values that would result in an effective speedup.

\begin{table}[h]
\begin{tabular}{|c  |c |c |c| c|}
\hline
$\alpha$  & N. Primes Used & Size of $\log$ Tables & Size of Reduction Tables& Circuit Size\\
\hline
0.062   & 9367 & 766 MB & 1 GB & $10^{15}$\\
0.5 & 1557697 & 49 TB & 280 GB &  $10^{18}$ \\
0.75   & 21462282 &13 PB  & 4 TB &  $10^{20}$\\
\hline
\end{tabular}
\caption{Memory requirements for $n=70034$.}
\label{tab:memcirc}
\end{table}

\subsection{Conclusion}
In this section, we analysed how to compute the modular exponentiation as efficiently as possible.
This analysis is important to estimate the relative speed between an average user and a resourceful attacker.
The speedup factor that one can achieve using top-tier commercial hardware and parallelisation seems to be bounded by a small integer.
In our discussion, we estimated a factor of roughly 20 compared to our setup when using a prime of 70034 bits.
This number is only indicative to show the non-prohibitive magnitude of possible speedups.
By using smaller primes, the chances of obtaining performance boosts via parallelisation decrease drastically. 
However, large enough primes should be used to discourage the construction of ad hoc boolean circuits  and other uses of lookup tables which, in theory, could pose a serious threat to our scheme.

\section{Post-quantum security}\label{sec:pq}
In this section, we discuss where the proposed scheme stands in a world where an adversary has access to large quantum computers.
In particular, we look at whether known quantum solutions to the  discrete logarithm problem could lead to potential attacks.

\subsection{Using a discrete logarithm oracle}
\label{sec:qatt}
In this subsection, we investigate how being able to compute the discrete logarithm may help break the delay of our scheme.
Let $p$ be a safe prime of $n$ bits. Let $c \coloneqq m^3 \mmod p$ and let $g$ be a generator of $\mathbb Z_p^*$.
The aim of an adversary is to obtain $m$ from $c$ in less than $n-1$ sequential modular multiplications.
We let such an adversary have access to an oracle that can compute the discrete logarithms in $\mathbb Z_p^*$.

Let $h = g^3$ and compute $x = \log _h (c)$. Then $g^x \equiv m \mod p$.
Assuming $x$ is effectively a random number between 1 and $p-1$, the probability that $x$ is $k$ bits long is
\[
\prob{ x \text{ is $k$ bits long}} =
\begin{cases}
\frac {2^{k-1}} {p-1} & \text{ if } k < n \\
\frac { p- 2^{n-1}} {p-1}  & \text{ if } k = n
\end{cases}
\]
Therefore, the expected bitsize of $x$ is $n-2 + \frac { 2(p - 2^{n-1}) + n -1} {p-1}$. 
So the probability that $x$ is low enough to save at least $N$ multiplications is $\frac {2^{n-1-N} -1}{p-1} \approx \frac 1 {2^N}$.
However, the base $g$ is known to the adversary before $x$, therefore they could precompute a table of values $g^{2^i} \mmod p$ for $i \in \{1, \dots, n-1\}$.
This would allow the computation of $g^x \mmod p$ to be organised in a binary tree and therefore require only $O(\log n)$ sequential multiplications.

In order to estimate the memory requirements, we fix $n=70034$. 
It follows that to store $g^{2^i} \mod p$ for all $i \in \{1, \dots, n-1\}$ one can use $n(n-1)$ bits which is roughly 585 MB.
As a result, we believe this attack is a viable option. 
Therefore, it is essential to understand whether an efficient logarithm oracle could be built in the real world as we understand it today.

\subsection{Computing the discrete logarithm}

\subsubsection{Classical algorithms for the discrete logarithm}\label{sec:classicalDlog}
The fastest classical algorithm to compute the discrete logarithm is the, so-called, number field sieve.
Despite being an exponential algorithm, Adrian \etal \cite{DHKE} show how enough precomputation can be used to compute the discrete logarithm of a DHKE group in roughly one minute.
We briefly explain the idea behind their work.
Let the DHKE protocol have parameters $p, g$ where $p$ is the prime modulo and $g$ the generator.
To compute $\log _ g b$ for some $b$, one first determines a number $l \in [1, p-1]$ such that $g^l b = q_1q_2\dots q_t$ where all $q_i$ are ``small''.
It follows that $\log _g b \mmod p$ can be deduced quickly from $\log _g q_i \mmod p$. 
Hence, the precomputation phase is spent on constructing a database of values $\log _g q_i \mmod p$. 
This precomputation can be parallelised, yet it still requires a lot of computational resources.
In the same paper, the authors estimate that to be able to compute the discrete logarithm in a 1024-bits DHKE group in almost ``real time'', the malicious entity would need to invest a dozen billion dollars and a couple of years of precomputation. 
Given that in our settings we work with much larger primes, we do not believe any classical attack would pose a threat even if the malicious party were a state entity.

\subsubsection{Eker\r{a} algorithm}
Eker\r{a} \cite{Qlog}  showed how one can improve Shor's algorithm by reducing the number of group operations needed.
The algorithm can be split in two parts: a quantum experiment that produces a set of pairs and a classical post-processing algorithm that yields the result.
The quantum experiment can be seen in Figure \ref{fig:ekera}.  The algorithm is used to compute $\log _g x \mmod p$.
\begin{figure}[h]
\procb{Quantum experiment}{
\pcinput x, g, p, s\\
\pcoutput j, k\\
\text{pick } m \text{ so that } 2^{m-1} \leq p < 2^m\\
l \gets \ceil{\frac m s}\\
a \gets \text{ superposition of } 0 \text{ to } 2^{m+l}-1\\
b \gets \text{ superposition of } 0 \text{ to } 2^{l}-1\\
c \gets g^ax^{-b} \mod p\\
\text{apply  quantum Fourier transform of size $2^{m+l}$ and $2^l$ to } (a,b,c)\\
\text{measure the result } (j,k,y)\\
\pcreturn j, k
}
\caption{Eker\r{a}'s quantum experiment.}
\label{fig:ekera}
\end{figure}
The crucial detail that allows Eker{\aa}'s algorithm to be useful in our setting is that the size of the quantum register $b$ can be bounded using the parameter $s$.
Note that the value of $g^a \mmod p$ can be computed before the receipt of the encrypted message $x$. 
Therefore, only the computation on $x$ and the Fourier transform must be performed during the ``delay period''.
Since the size of $b$ is $\frac m s$ bits, one could pick $s$ so that $x^{-b} \mmod p$ can be computed much faster than decrypting $x$.

The classical post-processing part of the algorithm can be seen in Figure \ref{fig:dlog}.
\begin{figure}[h]
\procb{Post-processing}{
\pcinput \set{(j_1, k_1), \dots, (j_n, k_n)}, m, s\\
\pcoutput d\\
l \gets \ceil{\frac m s}\\
f(x) \coloneqq (x\mmod 2^{m+l}) - 2^{m+l} \floor{ \frac {x \mmod 2^{m+l}} {2^{m+l-1}}}\\
v \gets \vector{f(-2^mk_1), \dots, f(-2^mk_n), 0}\\
L \gets \text{ lattice with basis } \set{(j_1, \dots, j_n , 1)} \cup \set{2^{m+l}e_i \middle|\;\; i \in \set{1, \dots, n}}\\
u \gets \text{ vector in $L$ closest to } v \\
\pcreturn u[n]\; (\text{last entry in the vector})
}
\caption{Eker\r{a}'s post-processing to compute the discrete logarithm.}
\label{fig:dlog}
\end{figure}
We report here a simplified version that suits our needs; the full algorithm can be seen in \cite{Qlog2}.
Computing $u$ from $v$ is an instance of the closest vector problem (CVP).
The CVP is considered infeasible in high dimensions, therefore the parameter $n$ should be chosen carefully.
The simplest algorithm to solve CVP is Babai's nearest plane algorithm {\cite{Babai}} (see Figure \ref{fig:babai}).
\begin{figure}[h]
\procb{Babai CVP}{
\pcinput \set{v_1, \dots, v_n}, t\\
\pcoutput x\\
\set{b_1, \dots, b_n} \gets \text{ LLL reduction of } \set{v_1, \dots, v_n}\\
\set{b_1^*, \dots, b_n^*} \gets \text{ Gram-Schmidt orthogonalisation of } \set{b_1, \dots, b_n}\\
x \gets t\\
\pcfor j \textbf{ from } n \textbf{ to } 1\colon\\
\t x \gets x - \left\lceil \frac{(x\cdot b^*_j)}{\norm{b^*_j}^2}\right\rfloor b_j\\
\pcreturn  t- x
}
\caption{Babai's nearest plane algorithm ($\left\lceil y\right\rfloor$ denotes the nearest integer to $y$).}
\label{fig:babai}
\end{figure}
Once the reduced basis and the Gram-Schmidt orthogonalisation are computed, the algorithm is relatively fast.
The loop requires $3n^2$ multiplications, which is reasonable as we aim to use small values for $n$.
The step that dominates the complexity of the algorithm is the LLL reduction of the basis. 
Nguyen and Sthelé \cite{LLL} propose a complex method to LLL reduce a lattice basis  in time $O(d^2n(d+\log B)M(d) \log B )$ where $d$ is the dimension of the lattice, $n$ the dimension of the underlying vector space, $B$ the largest norm of a vector in the basis and $M(d)$ indicates the time required to multiply two $d$-bits numbers. 
This algorithm suits our use well since the values of $d$ and $n$ are expected to be small.
In addition, this algorithm computes the values $\mu_{i,j} = \frac{ b_i \cdot b_j^*} {||b_j^*||^2}$ and $r_{i,i} = ||b_i^*||^2$ which can be used to speed up the computation of the Gram-Schmidt orthogonal basis.

As a result, we believe the post-processing part of Eker\r{a}'s algorithm could be computed in a small amount of time, assuming that $n$ is small.

Eker\r{a} \cite{Qlog2} analyses the relation between the parameters $s$ and $n$ and states that, when fixing a success probability of 99\%,``$n$ tends to $s+1$ as $m$ tends to infinity for a fixed $s$''.
This means that by picking $s=4$ we need roughly $5$ independent tests to compute the logarithm with a success probability of 99\%.
Since the tests are independent, they can be run in parallel on different quantum computers. 
Moreover, we expect the success probability to decrease linearly with $n$, so that using $s=4$ but $n=1$ would give roughly 25\% success probability.

Another recent paper by Gidney and Eker\r{a} {\cite{Qlogcost}}  analyses the real-world requirements for running quantum algorithms like the one described above.
In particular, they examined different quantum algorithms, among which Shor's and Eker\r{a}'s, in order to find the most efficient in terms of size and speed.
Since most of the quantum algorithms described in the literature are rather abstract like the one presented here, Gidney and Eker\r{a} also develop their own algorithm by piecing together optimisations from many sources. 
The first interesting result is that ``the cost of implementing the QFT [Quantum Fourier Transform] is negligible''{\cite{Qlogcost}}.
Moreover, Gidney and Eker\r{a} provided the Python script that was used to generate their estimate. 
Taking advantage of their code, we were able to produce some estimates for when the modulus has 65536 bits.
The quantum experiment requires $m + \lfloor\frac m s\rfloor $ modular multiplications to compute $g^a$ and $\lfloor \frac m s \rfloor$ to compute $x^{-b}$.
In Table \ref{tab:precomp} we show the estimated costs for the computation of $g^a$. 
In these calculations we minimise the size of the quantum computer, and we allow any percentage of failure as long as the expected runtime is less than a year.
The failure percentage shown in the table is the probability of the quantum computer to fail and it is not correlated with the success probability of Eker\r{a}'s algorithm.
To represent the complexity of the quantum algorithm, both its size in mega qubits and its Tiffoli count are listed. 
The expected runtime is nothing but the runtime weighted by the average number of trials needed for the computation to succeed.
We refer to \cite{Qlogcost} to deduce more parameters and note other assumptions such as physical gate error.
In Table \ref{tab:realtime1} we show the estimated costs for the computation of $x^{-b}$, where we minimise the single-run time.
Note that using the trade-off factor of $65536$ only one multiplication is needed.
However, the time required is roughly 2 minutes. 
Moreover, using such large values for $s$ is problematic as the probability of success of the post-processing is very low unless the quantum experiment is run multiple times.
In such a scenario, the post-processing algorithm becomes impractical as the parameter $n$ is no longer limited. 
As a result, in order to improve the running time, one needs to decrease the surface code cycle time and the reaction time of the machine.
Gidney and Eker\r{a} set these two parameters to the default values of  $1 \mu s$ and $10 \mu s$ respectively.
Therefore, it is unlikely that in the foreseeable future quantum computers will be efficient enough to pose a threat to our delay encryption.

\begin{table}[h]
\centering % centering table
\makebox[\textwidth]{
\begin{tabular} {|c|c|c | c|c|c|}
\hline
Trade-Off ($s$) &  Failure & Size (Mqb) & Runtime& Expected Runtime&  Tiffoli Count ($\times 10^9$)\\
\hline 
1      &$5\%$  &$3637$    &$8085$ h    & 355 days  &$148105$ \\
2         &$17\%$     &$3282$    &$7241$ h  & 364  days&$147957$ \\
4        &$13\%$   &$1975$     &$7412$ h   &  355 days&$58196$    \\
10        &$12\%$  &$1975$     &$6864$ h  &  325  days&$63764$    \\
100         &$12\%$  &$1975$     &$7127$ h  &  337 days&$73038$ \\
1000       &$12\%$   &$1975$     &$7063$ h   &  334  days&$72387$      \\
\hline
\end{tabular}}
\caption{cost estimates for the precomputation phase; we minimise the size of the quantum computer.}
\label{tab:precomp}
\end{table}

\begin{table}[h]
\centering % centering table
\makebox[\textwidth]{
\begin{tabular} {|c|c|c|c|c| c|c|c|c|c| c|c|c|}
\hline
Trade-Off ($s$) &  Failure & Size (Mqb) & Runtime& Expected Runtime&  Tiffoli Count ($\times 10^9$)\\
\hline 
1      &$80\%$  &$5128$    &$2425$ h    & 505 days  &$77387$ \\
2         &$42\%$     &$5128$    &$1211$ h  & 87 days &$38624$ \\
4        &$23\%$   &$5128$     &$605$ h   &  33 days &$19277$    \\
10        &$71\%$  &$4160$     &$237$ h  &  34  days &$7698$    \\
100         &$70\%$  &$3685$     &$23$ h  &  3 days&$766$ \\
1000       &$47\%$   &$1282$     &$2$ h   &  4 h  &$77$      \\
10000  & 61\% & $934$ & 14 min & 36 min & 6 \\
30000 & 31\% & 934 & 6 min & 9 min & 3 \\
65536 (max)& 59\% & 899& 2 min & 5 min & 1\\
\hline
\end{tabular}}
\caption{cost estimates to compute $x^{-b}$; we minimise the length of a single run.}
\label{tab:realtime1}
\end{table}

\section{Conclusion}
In this paper, we have proposed and analysed a delay encryption scheme based on cubing modulo a prime.
This represents the only alternative we are aware of to the original time-lock puzzle proposed by Rivest \textit{et al.} \cite{Rivest}
In particular, our scheme can be viewed as a variation of the RSA-based time-lock puzzle where we removed the hardness assumption of integer factorisation.
We performed an extensive analysis of potential attacks on our scheme and argued that they are not effective.
Among others, we have considered attacks using a quantum computer to solve the discrete logarithm problem, and we have reached the conclusion that no such attack would pose a threat in the foreseeable future.
Since our aim is to construct short delays, all of our concrete computations used a safe prime of 70034 bits. 
We believe that the same results will hold true for safe primes of the same magnitude, and these  represent good candidates for safe delays of around 1 second.
Thanks to the use of the chaining technique, one can construct delays that are secure and long enough for the use in communication protocols.
We also remark that concrete and exhaustive testing should be carried out before adopting this scheme in the wild.

\bibliographystyle{splncs04}
\bibliography{TREByCubing}

\end{document}